\documentclass[final,3p,times,twocolumn]{elsarticlehack}
\usepackage{graphicx}
\usepackage{amsthm}
\usepackage[utf8]{inputenc}
\usepackage[fleqn]{amsmath}
\usepackage{xspace,makecell}
\usepackage{bm,fullpage}
\usepackage{amssymb}
\usepackage[lined,boxed]{algorithm2e}

\usepackage{todonotes}
\usepackage{hyperref}

\DeclareMathOperator{\tw}{tw}

\def\CC{\mathcal{C}}
\def\FF{\mathcal{F}}
\def\HH{\mathcal{H}}
\def\KK{\mathcal{K}}

\newcommand{\FO}{{\sf FO}\xspace}
\newcommand{\MSO}{{\sf MSO}\xspace}
\newcommand{\MSOJ}{{\sf MSO$_{1}$}\xspace}
\newcommand{\MSOD}{{\sf MSO$_{2}$}\xspace}
\def\prebox#1{\mathop{\mbox{\rm #1}}}

\def \ew {c}

\def\R{\mathbb{R}}
\def\NN{\mathbb{N}}

\def\gc{\ensuremath{{\cal S}_G}\xspace}
\def\gcmc{\ensuremath{\mathsf{GCMC}}\xspace}

\def\sse{\subseteq}

\newcommand{\Feas}{{\rm Feas}\xspace}

\newtheorem{theorem}{Theorem}[section]
\newtheorem{lemma}{Lemma}
\newtheorem{cl}{Claim}
\newtheorem{observation}{Observation}
\newtheorem{assume}{Assumption}
\theoremstyle{remark}
\newtheorem{define}{Definition}
\newtheorem{remark}{Remark}

\newcounter{note}[section]

\newcommand{\ignore}[1]{}

\begin{document}

%

\title{Approximating Max-Cut under Graph-MSO Constraints}


\author[KAM]{Martin Koutecký}
\ead{koutecky@kam.mff.cuni.cz}
\author[UMich]{Jon Lee}
\ead{jonxlee@umich.edu}
\author[UMich]{Viswanath Nagarajan}
\ead{viswa@umich.edu}
\author[UMich]{Xiangkun Shen}
\ead{xkshen@umich.edu}

\address[KAM]{Technion -- Israel Institute of Technology, Faculty of IE\&M, Haifa, Israel}
\address[UMich]{IOE Dept., University of Michigan, Ann Arbor, MI 48109, USA}

\begin{abstract}
We consider the max-cut and max-$k$-cut problems under graph-based constraints. Our approach can handle any constraint specified using monadic second-order (MSO) logic on graphs of constant treewidth. We give a $\frac12$-approximation algorithm for this class of  problems.
\end{abstract}




\maketitle

\section{Introduction}\label{sec:intro}

\def\cons{\ensuremath{{\cal S}_G}\xspace}

This paper considers the classic max-cut problem under a class of graph-based constraints. The max-cut problem is a fundamental combinatorial-optimization problem which has many practical applications (see \cite{CD87,BGJR88,JS12,LRSST10}) as well as strong theoretical results (see \cite{GW95,KMO07}). There have also been a number of papers on designing approximation algorithms for {\em constrained} max-cut problems (see \cite{AHS01,AS99,CVZ14,FNS11,HKMPS15}).

In this paper, we are interested in constraints that are specified by an auxiliary constraint graph. Our main result is a $\frac12$-approximation algorithm for max-cut under any graph constraint that can be expressed in monadic second order logic (\MSO) (see \cite{Courcelle:1996}). This is closely related to a recent result by a subset of the authors; see \cite{SLN17}. The contribution of this paper is in
generalizing the class of constraints handled in \cite{SLN17}, making the algorithm design more systematic, and extending the result to the max-$k$-cut setting with $k$ instead of just $2$ parts.

In particular, \cite{SLN17} gave a $\frac12$-approximation algorithm for max-cut under any graph constraint $\gc$ that has a specific type of dynamic program for optimizing linear objectives. In order to apply this result, one also has to design such a dynamic program separately for each constraint $\gc$, which requires additional constraint-specific work. Indeed, \cite{SLN17} also gave constraint-specific dynamic programs for various graph constraints such as independent set, vertex cover, dominating set and connectivity, all on bounded-treewidth graphs.

In this paper, we  bypass the need for constraint-specific dynamic programs by utilizing the language and results from monadic second-order logic. We show that any \MSO constraint on a bounded-treewidth graph (defined formally in \S\ref{sec:prelim}) admits a dynamic program that satisfies the assumptions needed in \cite{SLN17}. Therefore, we immediately obtain $\frac12$-approximation algorithms for max-cut under any \MSO graph constraint. We note that \MSO constraints capture all the specific graph constraints in \cite{SLN17}, and much more.

We also extend these results to the setting of max-$k$-cut, where we seek  to partition the vertices into $k$ parts $\{U_i\}_{i=1}^k$
so as to maximize the weight of edges crossing the partition. In the constrained version, we additionally require each part $U_i$ to satisfy some \MSO graph property. We obtain a $\frac12$-approximation algorithm even in this setting ($k$ is fixed). This result is a significant generalization over \cite{SLN17} even for $k=2$, which corresponds to the usual max-cut problem: we now handle constraints on both sides of the cut.

\section{Preliminaries}\label{sec:prelim}
A $k$-partition of  vertex set $V$ is a function $h:V\rightarrow [k]$, where the $k$ parts are $U_\alpha = \{v\in V: h(v)=\alpha\}$ for $\alpha\in [k]$. Note that $\cup_{\alpha=1}^k U_\alpha = V$ and  $U_1,\cdots, U_k$ are disjoint. When we want to refer to the $k$ parts directly, we also use $\{U_\alpha\}_{\alpha=1}^k$ to denote the $k$-partition.

\begin{define}[\gcmc]
The input to the {\em graph-constrained max-cut} (\gcmc) problem consists of (i) an $n$-vertex graph $G=(V,E)$ with a graph property which implicitly specifies a collection ${\cal S}_G$ of vertex $k$-partitions,
and (ii) symmetric edge-weights $\ew:{V\choose 2} \rightarrow \R_+$. The \gcmc problem is to find a $k$-partition in \gc with the maximum weight of crossing edges:
\begin{equation}\label{eq:gcmc-defn}
\max_{h \in \gc} \quad \sum_{\substack{\{u, v\} \in {V\choose 2}  \\ h(u)\ne h(v)}} \ew(u,v).
\end{equation}
\end{define}

\paragraph{Tree Decomposition}
Given an undirected graph $G=(V,E)$, a tree decomposition  consists of a tree ${\cal T}=(I,F)$ and a collection of vertex subsets $\{X_i\subseteq V\}_{i\in I}$
such that:
\begin{itemize}
\setlength\itemsep{0pt}
	\item for each $v\in V$, the nodes $\{i \in I: v\in X_i\}$ are connected in ${\cal T}$, and
	\item for each edge $(u,v)\in E$, there is some node $i\in I$ with $u,v\in X_i$.
\end{itemize}

The width of such a tree decomposition is $\max_{i\in I}(|X_i|-1)$, and the treewidth of $G$ is the smallest width of any tree decomposition for $G$.

We  work with ``rooted'' tree decompositions, also specifying a root node $r\in I$. The depth $d$ of such a tree decomposition is the length of the longest root-leaf path in ${\cal T}$. The depth of any node $i\in I$ is the length of the $r-i$ path in ${\cal T}$. For any $i\in I$, the set $V_i$ denotes all the vertices at or below node $i$, that is
$$V_i:= \cup_{k \in {\cal T}_i}  X_k,$$
$$\mbox{where }{\cal T}_i=\{ k\in I : k \mbox{ in subtree of ${\cal T}$ rooted at }i\}.$$

The following result provides a convenient representation of ${\cal T}$.

\begin{theorem}[Balanced Tree Decomposition; see \cite{Bodlaender88}]
	Let $G=(V,E)$ be a graph with tree decomposition $({\cal T}=(I,F), \{X_i|i\in I\})$ of treewidth $k$. Then $G$ has a rooted tree decomposition $({\cal T'}=(I',F'), \{X'_i|i\in I'\})$ where ${\cal T'}$ is a binary tree of depth $2\lceil \log_\frac{5}{4} (2|V|)\rceil$ and treewidth at most $3k+2$. Moreover, for all $i\in I$, there is an $i'\in I'$ such that $X_i\subseteq X'_{i'}$.  The tree decomposition ${\cal T'}$ can be found in $O(|V|)$ time.
	\label{thm:treedecomp}
\end{theorem}

\begin{define}[CSP instance]
	 A Constraint Satisfaction Problem (CSP) instance $J=(N,\CC)$ consists of:
	\begin{itemize}
\setlength\itemsep{0pt}
		\item a set $N$ of {\em boolean variables}, and
		\item a set $\CC$ of {\em constraints},
		where each
		constraint $C_{U} \in \CC$
		is a $|U|$-ary relation $C_U \subseteq \{0,1\}^U$ on some subset $U\subseteq N$.
	\end{itemize}
\end{define}

For a vector $x\in \{0,1\}^N$ and a subset $R$ of variables, we denote by $x|_{R}$ the \emph{restriction of $x$ to $R$}.
A vector $z\in \{0,1\}^N$ \emph{satisfies constraint} $C_U \in \CC$  if $z|_U \in C_U$.
We say that  $z\in \{0,1\}^N$
is \emph{a feasible assignment} for the CSP instance
$J$ if  $z$ satisfies every constraint $C\in \CC$. Let $\Feas(J)$ be the set of all feasible assignments of $J$.
Finally,  $\|\CC\| = \sum_{C_U \in \CC} |C_U|$ denotes the \emph{length of $\CC$}.

\begin{define}[Constraint graph]
	The  constraint graph of $J$,  denoted $G(J)$, is defined as $G(J)=(N,F)$
	where $F= \left\{\{u,v\} \ \mid\  \exists C_{U} \in \CC \textrm{ s.t. }
	\{u,v\} \subseteq U\right\}$.
\end{define}

\begin{define}[Treewidth of CSP]
	The  treewidth  $\tw(J)$ of a CSP instance $J$ is defined as the treewidth of its constraint graph $\tw(G(J))$.
\end{define}


\begin{define}[CSP extension]
	Let $J = (N, \CC)$ be a CSP instance.
	We say that $J' = (N', \CC')$ with $N \subseteq N'$ is an extension of $J$ if $\Feas(J) = \big\{z|_{N} ~\big|~ z \in \Feas(J')\big\}$.
\end{define}

\paragraph{Monadic Second Order Logic}

We briefly introduce \MSO over graphs.
In {\em first-order logic} (\FO) we have variables for individual vertices/edges
(denoted $x,y,\ldots$), equality for variables, quantifiers $\forall,\exists$
ranging over variables, and the standard Boolean connectives $\neg, \wedge, \vee, \implies$.
\MSO is the extension of \FO by quantification over sets (denoted $X,Y,\dots$).
Graph \MSO has the binary relational symbol $\prebox{edge}(x,y)$ encoding edges, and traditionally comes in two
flavours, \MSOJ and \MSOD, differing by the objects we are allowed to
quantify over:
in \MSOJ these are the vertices and vertex sets,
while in \MSOD we can additionally quantify over edges and edge sets.
For example, $3$-colorability can be expressed in \MSOJ as follows:
\begin{eqnarray*}
	\exists X_1,X_2,X_3 &&\left[\,
	\forall x \, (x\in X_1\vee x\in X_2\vee x\in X_3)  \right.
	\\ &&\wedge\bigwedge\nolimits_{i=1,2,3} \left. \!\!
	\forall x,y
	\left(x\not\in X_i\vee y\not\in X_i\right.\right.\\	
	&&\left.\left.\vee\neg\prebox{edge}(x,y)\right)
	\,\right]
\end{eqnarray*}
We  remark that \MSOD can express properties that are not \MSOJ definable.
As an example, consider Hamiltonicity on graph $G=(V,E)$; an equivalent description of a Hamiltonian cycle is that it is a connected $2$-factor of a graph:
\begin{align*}
&\varphi_{\text{ham}} \equiv \exists F \subseteq E: \, \varphi_{\text{2-factor}}(F) \wedge \varphi_{\text{connected}}(F) \\
&\varphi_{\text{2-factor}}(F)  \equiv (\forall v \in V: \, \exists e,f \in F: \, (e\neq f)\\
&\qquad \wedge (v \in e) \wedge (v \in f) ) \wedge \neg ( \exists v \in V: \,\\
&\qquad \exists e,f,g \in F: (e\neq f \neq g) \wedge (v \in e)\\
&\qquad \wedge (v \in f) \wedge (v \in g) ) \\
&\varphi_{\text{connected}}(F)  \equiv \neg \big[  \exists U,W \subseteq V:\, (U \cap W = \emptyset)\\
&\qquad \wedge (U \cup W = V)  \wedge \neg \big( \exists \{u,v\} \in F:\\
&\qquad u \in U \wedge v \in W \big) \big] \enspace .
\end{align*}

We use $\varphi$ to denote an \MSO formula and $G=(V,E)$ for the underlying graph. For a formula $\varphi$, we denote by $|\varphi|$ the \emph{size} (number of symbols) of $\varphi$.

In order to express constraints on $k$-vertex-partitions via \MSO, we use \MSO formulas $\varphi$ with $k$ free variables $\{U_\alpha\}_{\alpha=1}^k$ where (i) the $U_\alpha$ are enforced to form a partition of the vertex-set $V$, and (ii) each $U_\alpha$ satisfies some individual  \MSO constraint $\varphi_\alpha$. Because $k$ is constant, the size of the resulting \MSO formula is a constant as long as each of the \MSO constraints $\varphi_\alpha$ has constant size.

\paragraph{Connecting CSP and \MSO}

Consider an \MSO formula $\varphi$ with $k$ free variables on graph $G$ (as above). For a vector $t\in \{0,1\}^{V\times [k]}$, we write
$G,t \models \varphi$ if and only if $\varphi$ is satisfied by solution $U_\alpha=\{v\in V: t(({v,\alpha})) =1 \}$ for $\alpha\in [k]$.

\begin{define}[$CSP_\varphi(G)$ instance]\label{def:CSP_G}
	Let $G$ be a graph and $\varphi$ be an \MSOD-formula with $k$ free variables.
	By $CSP_\varphi(G)$ we denote the CSP instance $(N, \CC)$ with $N = \{t(({v,\alpha})) \mid v \in V(G), \alpha\in [k]\}$ and with a single constraint $\{t \mid G,t \models \varphi\}$.
\end{define}

Observe that $\Feas(CSP_\varphi(G))$ corresponds to the set of feasible assignments of $\varphi$ on $G$. Also, the treewidth of $CSP_\varphi(G)$ is $|V|k$ which is unbounded. The following result shows that there is an equivalent CSP extension that has constant treewidth.

\begin{theorem}[{\cite[Theorem 25]{KnopKMT:2017}}]\label{thm:MSO_CSP}
	Let $G=(V,E)$ be a graph with $\tw(G) = \tau$ and $\varphi$ be an \MSOD-formula with $k$ free variables.
	Then $CSP_\varphi(G)$ has a CSP extension $J$ with $\tw(J) \leq f(|\varphi|, \tau)$ and  $\|\CC_J\| \leq f(|\varphi|, \tau) \cdot |V|$.
\end{theorem}
To be precise,~\cite[Theorem 25]{KnopKMT:2017} speaks of \MSOJ over $\sigma_2$-structures, which is equivalent to \MSOD over graphs; cf. the discussion in~\cite[Section 2.1]{KnopKMT:2017}.
\ignore{
\begin{theorem}[{Lemma 5.2.8~\cite{KouteckyDiss}}]
	Let $G$ be a graph with $\tw(G) = \tau$, $\varphi$ be an \MSOD-formula with $k$ free variables.
	Then $CSP_\varphi(G)$ has a CSP extension $J$ with $\tw(J) = 2$, $\|\CC_J\| \leq f(|\varphi|, \tau) \cdot |V(G)|$, and $D_J \leq f(|\varphi|, \tau)$.
\end{theorem}}
\section{Dynamic Program for CSP}
In this section we demonstrate that every CSP of bounded treewidth admits a dynamic program that satisfies the assumptions required in \cite{SLN17}.


Consider a CSP instance $J=(V,\CC)$ with a constraint graph $G=(V,E)$ of bounded treewidth. Let  $({\cal T}=(I,F), \{X_i|i\in I\})$  denote a balanced tree decomposition of $G$ (from Theorem~\ref{thm:treedecomp}).
In what follows, we denote the vertex set $V=[n]=\{1,2,\cdots, n\}$.
Let $\lambda$ be a symbol denoting an unassigned value. 
 For any  $W \subseteq V$,  define the \emph{set of
 configurations of} $W$ as:
\begin{align*}
\KK(W) = &\big\{(z_1, \dots, z_n) \in \{0,1,\lambda\}^V\ |\\
& \forall C_U\in \CC:  (U\subseteq W \implies z|_{U}\in C_U),\\ &
\forall i \not\in W: z_i = \lambda,~
\forall j\in W: z_j \in \{0,1\} \big\}
\end{align*}
Let $k \in \KK(W)$ be a configuration and $v\in V$.
Because $k$ is a vector, $k(v)$ refers to the $v$-th element of $k$.

\begin{define}[State Operations]
	Let $U,W\sse V$. Let $k\in \KK(U)$ and $p\in \KK(W)$. 
	\begin{itemize}
\setlength\itemsep{0pt}
		\item Configurations $k$ and $p$ are said to be {\em consistent} if, for each $v\in V$, either $k(v)=p(v)$ or at least one of $k(v), p(v)$ is $\lambda$.
		\item If configurations $k$ and $p$ are consistent, define $[p\cup k](v) = \begin{cases}
		p(v),\ &\mbox{if }k(v)=\lambda;\\
		k(v),\ &\mbox{otherwise.}
		\end{cases}$
		\item Define $[k\cap W](v)=\begin{cases}
		k(v),\ &\mbox{if }v\in W;\\
		\lambda,\ &\mbox{otherwise.}
		\end{cases}$
	\end{itemize}
\end{define}

We start by defining some useful parameters for the dynamic program.

\begin{define} \label{def:DP}
For each node $i\in I$ with children nodes $\{j,j'\}$, we associate the following:
	\begin{enumerate}
		\item  state space $\Sigma_i = \KK(X_i)$.
		\item for each  $\sigma\in\Sigma_i$, there is a collection of partial solutions
		$$\HH_{i,\sigma}  := \{k \in \KK(V_i) \mid k \cap X_i = \sigma\}  .$$
		\item for each  $\sigma\in\Sigma_i$, there is a collection  of valid combinations of children states
		\begin{align*}
		\FF_{i, \sigma} = &\{(\sigma_j, \sigma_{j'}) \in \Sigma_j \times \Sigma_{j'} \mid (\sigma_j \cap X_i) =\\& (\sigma \cap X_j) \text{ and } (\sigma_{j'} \cap X_i) = (\sigma\cap X_{j'}) \}.
		\end{align*}
	\end{enumerate}
\end{define}
In words, (a) $\Sigma_i$ is just the set of configurations for the vertices $X_i$ in node $i$, (b)
$\HH_{i, \sigma}$ are those configurations for the vertices $V_i$ (in the subtree rooted at $i$) that are  consistent with $\sigma$, (c) $\FF_{i, \sigma}$ are those pairs of states at the children $\{j,j' \}$ that agree with $\sigma$ on the intersections $X_i\cap X_j$ and $X_i\cap X_{j'}$ respectively.

\begin{theorem}[{\small Dynamic Program for CSP}]\label{thm:DP}
	Let $({\cal T}=(I,F), \{X_i|i\in I\})$ be a tree decomposition of a CSP instance $(V,\CC)$ of bounded treewidth.
	Then $\Sigma_i$, $\FF_{i,\sigma}$ and $\HH_{i,\sigma}$ from Definition~\ref{def:DP} satisfy the conditions:
	\begin{enumerate}
		\item \label{asm:DP:state}\emph{(bounded state space)} $\Sigma_i$ and ${\FF}_{i,\sigma}$ are all bounded by constant, that is, $\max_i{|\Sigma_i|}=O(1)$ and $\max_{i,\sigma}{|{\FF}_{i,\sigma}|}=O(1)$.
		\item\label{asm:DP:req} \emph{(required state)} For each $i\in I$ and $\sigma\in\Sigma_i$, the intersection with $X_i$ of every vector in $\HH_{i,\sigma}$ is the same, in particular
		$h\cap X_i = \sigma$  for all $h\in \HH_{i,\sigma}$.
		\item \label{asm:DP:leaf} By condition~2, for any leaf $\ell\in I$ and $\sigma\in\Sigma_\ell$, we have $\HH_{\ell,\sigma}=\{\sigma\}$ or $\emptyset$.
		\item \label{asm:DP:valid} \emph{(subproblem)} For each non-leaf node $i\in I$ with children $\{j,j'\}$ and $\sigma\in\Sigma_i$,
		\begin{align*}
		\HH_{i,\sigma}= &\Big\{\sigma\cup h_j\cup h_{j'} \mid h_j\in\HH_{j,w_j},\\&\ h_{j'} \in\HH_{j',w_{j'}}, \, (w_j,w_{j'}) \in{\FF}_{i,\sigma}\Big\}.
		\end{align*}
		\item \label{asm:DP:root} \emph{(feasible subsets)} At the root node $r$, we have $\Feas(V,\CC) =\bigcup_{\sigma\in\Sigma_r}{\HH_{r,\sigma}}$.
	\end{enumerate}
\end{theorem}
\begin{proof}
Let $q=O(1)$ denote the treewidth of ${\cal T}$. We now prove each of the claimed properties.

\smallskip\noindent\emph{Bounded state space.} Because $|X_i| \leq q+1$, we have $|\Sigma_i| = |\KK(X_i)| \leq 3^{q+1} = O(1)$ and $|\FF_{i, \sigma}| \leq |\Sigma_{j} \times \Sigma_{j'}| \leq (3^{q+1})^2= O(1)$.

\smallskip\noindent\emph{Required state.} This holds immediately by definition of $\HH_{i, \sigma}$ in Definition~\ref{def:DP}.

\smallskip\noindent\emph{Subproblem.} We first prove the ``$\subseteq$'' inclusion of the statement.
Consider any $h \in \HH_{i, \sigma} \subseteq \KK(V_i)$. Let $h_j = h \cap V_j$, $w_j = h \cap X_j$ and analogously for $j'$.
Observe that for $U \subset W \subseteq V$ we have that $k \in \KK(W) \implies k \cap U \in \KK(U)$.
By this observation, $h_j \in \KK(V_j)$.
Moreover, $h_j\cap X_j = h\cap X_j = w_j$, which implies  $h_j \in \HH_{j, w_j}$.
Again, the same applies for $j'$ and we have $h_{j'} \in \HH_{j', w_{j'}}$. Finally, note that $w_j\cap X_i = h\cap X_j \cap X_i = (h \cap X_i) \cap X_j = \sigma\cap X_j$ and similarly $w_{j'}\cap X_i = \sigma\cap X_{j'}$. So we have $(w_j,w_{j'}) \in  \FF_{i, \sigma}$.

Now, we prove the
``$\supseteq$'' inclusion of the statement.
Consider any two partial solutions $h_j \in \HH_{j, w_j}$ and $h_{j'} \in  \HH_{j', w_{j'}}$ with $(w_j, w_{j'}) \in \FF_{i, \sigma}$. Note that $h_j$ and $\sigma$ (similarly $h_{j'}$ and $\sigma$) are consistent by definition of $\FF_{i, \sigma}$. We now claim that $h_{j}$ and $h_{j'}$  are also consistent: take any $v\in V$
with both $h_{j'}(v),  h_{j'}(v)  \ne \lambda$, then we must have $v\in V_j\cap V_{j'} \sse X_i\cap X_j\cap X_{j'}$ as $X_i$ is a vertex separator, and so $h_j(v) = \sigma(v) = h_{j'}(v)$ by definition of $\FF_{i, \sigma}$. Because $\sigma$, $h_j$ and $h_{j'}$ are mutually consistent, $h=\sigma\cup h_j\cup h_{j'} $ is well-defined. It is clear from the above arguments that $h\cap X_i = \sigma$. In order to show $h\in
\HH_{i, \sigma} $ we now only need  $h \in \KK(V_i)$, that is, $h$ does not violate any constraint that is contained  in $V_i$.
For contradiction assume that that there is such a violated constraint $C_S$ with
$S \subseteq V_i$.
Then $S$ induces a clique in the constraint graph $G$ and thus there must exist a node $k$ among the descendants of $i$ such that $S \subseteq V_k$.
But $k$ cannot be in the subtree rooted in $j$  or $j'$, because then $C_S$ would have been violated already in $h_j$ or $h_{j'}$, and also it cannot be that $i=k$, because then $C_S$ would be violated in $\sigma$, a contradiction.

\smallskip
\noindent\emph{Feasible subsets.} Clearly, the set $\Feas(V,\CC)$  of feasible CSP solutions  is equal to $\KK(V)$.
Because  $\HH_{r, \sigma}$ is  those $k \in \KK(V)$ with $k \cap X_r = \sigma$,  the claim follows.
\end{proof}

We note that Theorem~\ref{thm:DP} proves Assumption~1 in \cite{SLN17}. To clarify the comparison, Assumption~1 is:

\begin{assume}[Assumption 1 in~\cite{SLN17}]\label{assm:DP}
	Let $({\cal T}=(I,F), \{X_i|i\in I\})$ be any  tree decomposition. Then there exist $\Sigma_i$, ${\cal F}_{i,\sigma}$ and $\HH_{i,\sigma}$ (see Definition~\ref{def:DP}) that satisfy the following conditions:
	\begin{enumerate}
		\item \emph{(bounded state space)} $\Sigma_i$ and ${\cal F}_{i,\sigma}$ are all bounded by constant, that is, $\max_i{|\Sigma_i|}=O(1)$ and $\max_{i,\sigma}{|{\cal F}_{i,\sigma}|}=O(1)$.
		\item \emph{(required state)} For each $i\in I$ and $\sigma\in\Sigma_i$, the intersection with $X_i$ of every set in $\HH_{i,\sigma}$ is the same, denoted $X_{i,\sigma}$, that is $S\cap X_i=X_{i,\sigma}$ for all $S\in\HH_{i,\sigma}$.
		\item By condition~2, for any leaf $\ell\in I$ and $\sigma\in\Sigma_\ell$, we have $\HH_{\ell,\sigma}=\{X_{\ell,\sigma}\}$ or $\emptyset$.
		\item  \emph{(subproblem)} For each non-leaf node $i\in I$ with children $\{j,j'\}$ and $\sigma\in\Sigma_i$,
		\begin{align*}
		\HH_{i,\sigma}\quad =\quad &\Big\{X_{i,\sigma}\cup S_j\cup S_{j'} \,\,:\,\,S_j\in\HH_{j,w_j},\\
		&S_{j'} \in\HH_{j',w_{j'}}, \, (w_j,w_{j'}) \in{\cal F}_{i,\sigma}\Big\}.
		\end{align*}
		\item  \emph{(feasible subsets)} At the root node $r$, we have $\gc =\bigcup_{\sigma\in\Sigma_r}{\HH_{r,\sigma}}$.
	\end{enumerate}
\end{assume}
Assumption~1 is used in the main result of~\cite{SLN17},  which is restated below.
\begin{theorem}[Theorem 4 in~\cite{SLN17}]\label{thm:sym}
	Consider any instance of the \gcmc problem on a bounded-treewidth graph $G$. 
If the graph constraint \gc satisfies Assumption~\ref{assm:DP} then we obtain a $\frac12$-approximation algorithm.
\end{theorem}
We will use this result 
in Section~\ref{sec:maxcut}, but we will modify its proof slightly in Section~\ref{sec:maxcut-k} for max-$k$-cut.

 \section{The Max-Cut Setting}\label{sec:maxcut}
Here, we consider the \gcmc problem when $k=2$ and there is a constraint \gc for only one side of the cut. We show that the above dynamic-program structure can be combined with  \cite{SLN17} to obtain a $\frac12$-approximation algorithm. 

 Formally, there is an \MSO formula $\varphi$ with one free variable defined on graph $G=(V,E)$ of bounded treewidth. The feasible vertex subsets $\gc$ are  those $S\sse V$ that satisfy $\varphi$. There is also a symmetric weight function $\ew:{V\choose 2} \rightarrow \R_+$. We are interested in the following problem ($\gcmc_I$).
\begin{equation}\label{eq:gcmc2-defn}
\max_{S\in \gc} \quad \sum_{u\in S,\, v\not\in S} \ew(u,v).
\end{equation}
We note that this is precisely the setting of \cite{SLN17}.
\begin{theorem}\label{thm:k=2}
There is a $\frac12$-approximation algorithm for $\gcmc_I$ when the constraint \gc is given by any \MSO formula on a bounded-treewidth graph.
\end{theorem}
\begin{proof}
The proof uses Theorem~\ref{thm:sym} from \cite{SLN17} as a black-box.  
Note that the constraint \gc corresponds to feasible assignments to $CSP_\varphi(G)$ as in Definition~\ref{def:CSP_G}.
Consider the CSP extension $\vartheta$ obtained after applying Theorem~\ref{thm:MSO_CSP} to $CSP_\varphi(G)$.  Then $\vartheta$ has variables $V'\supseteq V$ and bounded treewidth. We obtain an extended  weight function $\ew:{V' \choose 2} \rightarrow \R_+$  from $\ew$  by setting $\ew'(u,v) = \ew(u,v)$ if $u,v\in V$ and $\ew'(u,v)=0$ otherwise. We now consider a new instance of $\gcmc_I$ on vertices $V'$ and constraint $\vartheta$. Due to the bounded-treewidth property of $\vartheta$, we can apply
Theorem~\ref{thm:DP} which proves that Assumption~\ref{assm:DP} is satisfied by the dynamic program in Definition~\ref{def:DP}. Combined with Theorem~\ref{thm:sym}, we obtain the claimed result.
\end{proof}

\section{The Max-$k$-Cut Setting}\label{sec:maxcut-k}
In this section, we generalize the setting to any constant $k$, i.e. problem \eqref{eq:gcmc-defn}. Recall the  formal definition from \S\ref{sec:prelim}. Here the graph property \gc is expressed as an \MSO formula with $k$ free variables on graph $G$. Our main result is the following:

\begin{theorem}\label{thm:gen-k}
There is a $\frac12$-approximation algorithm for any $\gcmc$ instance with  constant $k$ when the constraint \gc is given by any  \MSO formula on a bounded-treewidth graph.
\end{theorem}

\begin{remark} \label{rem:pc}
The complexity of Theorem~\ref{thm:gen-k} in terms of the treewidth $\tau$, length $|\varphi|$ of $\varphi$, depth $d$ of a tree decomposition of $G$, and maximum degree $r$ of a tree decomposition of $G$, is $s^{dr}$, where $s$ is the number of states of the dynamic program, namely $f(|\varphi|, \tau)$ for $f$ from Theorem~\ref{thm:MSO_CSP}.
From the perspective of parameterized complexity~\cite{DF} our algorithm is an \emph{XP algorithm} parameterized by $\tau$, i.e., it has runtime $n^{g(\tau)}$ for some computable function $g$.
\end{remark}

Let $G=(V,E)$ be the input graph (assumed to have bounded treewidth) and $\varphi$ be any \MSO formula with $k$ free variables.  Recall the CSP instance $CSP_\varphi(G)$ on variables $\{y({v,\alpha}) : v\in V, \alpha\in [k]\}$ from  Definition~\ref{def:CSP_G}. Feasible solutions to $CSP_\varphi(G)$ correspond to feasible $k$-partitions in \gc. Now consider the CSP extension $\vartheta$ obtained after applying Theorem~\ref{thm:MSO_CSP} to $CSP_\varphi(G)$.  Note that $\vartheta$ is defined on variables $V'\supseteq \{ (v,\alpha) : v\in V, \alpha\in [k] \}$ and has bounded treewidth. Let ${\cal T}$ denote the tree decomposition for $\vartheta$. Below we utilize the dynamic program from Definition~\ref{def:DP} applied to $\vartheta$: recall the quantities $\Sigma_i$, $\FF_{i, \sigma}$ etc. We will also refer to the variables in $V'$ as vertices, especially when referring to the tree decomposition ${\cal T}$; note that these are different from the vertices $V$ in the original graph $G$. 

\begin{cl}\label{cl:DP-sol-state}
Let $\{U_\alpha\}_{\alpha=1}^k$ be a $k$-partition satisfying \gc.
There is a collection of states $\{b[i]\in \Sigma_i\}_{i\in I}$ such that:
	\begin{itemize}
\setlength\itemsep{0pt}
		\item for each node $i\in I$ with children $j$ and $j'$, $(b[j],b[j']) \in {\FF}_{i,b[i]}$,
		\item for each leaf $\ell$ we have $\HH_{\ell,b[\ell]}\ne \emptyset$, and
		\item $U_\alpha = \{v\in V :  b_{\mathcal{T}}((v,\alpha))=1\}$ for all $\alpha\in [k]$, where  $b_{\mathcal{T}}=\bigcup_{i\in I} b[i]$.
	\end{itemize}
	Moreover, for any vertex $(v,\alpha)\in V'$, if $\overline{ v \alpha }\in I$ denotes the highest node in ${\cal T}$ containing $(v,\alpha)$ then we have: $v\in U_\alpha$  if and only if $b[\overline{ v \alpha }]((v,\alpha))=1$.
\end{cl}
\begin{proof}
By definition of CSP $\vartheta$, we know that it has some feasible solution $t \in \{0,1\}^{V'}$ where $U_\alpha = \{v\in V :  t({(v,\alpha)}) =1\}$ for all $\alpha\in [k]$. Now, using Theorem~\ref{thm:DP}(5) we have $t\in \bigcup_{\sigma\in\Sigma_r}{\HH_{r,\sigma}}$.
	We define the states $b[i]$ in a top-down manner. We will also define an associated vector $t_i\in \HH_{i,b[i]}$ at each node $i$. At the root, we set $b[r]=\sigma$ such that $t\in \HH_{r,\sigma}$: this is well-defined because $t \in \bigcup_{\sigma\in\Sigma_r}{\HH_{r,\sigma}}$. We also set $t_r=t$. Having set $b[i]$ and $t_i\in \HH_{i,b[i]}$ for any node $i\in I$ with children $\{j,j'\}$, we use Theorem~\ref{thm:DP}(\ref{asm:DP:valid})  to write
$h_i=b[i] \cup h_j\cup h_{j'}$  where $h_j\in \HH_{j,w_j}$, $h_{j'}\in \HH_{j',w_{j'}}$  and $(w_j,w_{j'}) \in {\cal F}_{i,b[i]}$.
	Then we set $b[j]=w_j$,  $t_j=h_j$ and $b[{j'}]=w_{j'}$,  $t_{j'}=h_{j'}$  for the children  of node $i$. The first condition in the claim is immediate from the definition of states $b[i]$. By induction on the depth of node $i$, we obtain $t_i\in \HH_{i,b[i]}$ for each node $i$. This implies that $\HH_{\ell,b[\ell]}\ne \emptyset$ for each leaf $\ell$, which proves the second condition; moreover, by Theorem~\ref{thm:DP}(\ref{asm:DP:leaf}) we have $t_\ell =  b[\ell]$. Now, by definition of the vectors $t_i$, we obtain $t=t_r = \bigcup_{i\in I} b[i] =  b_{\mathcal{T}}$ which, combined with   $U_\alpha = \{v\in V :  t(({v,\alpha})) =1\}$ for all $\alpha\in [k]$, proves the third condition in the claim.
	
	Because $t=\bigcup_{i\in I} b[i]$, it is clear that if $b[\overline{ v \alpha }]((v,\alpha))=1$ then $v\in U_{\alpha}$. In the other direction, suppose $b[\overline{ v \alpha }]((v,\alpha))\ne1$: we will show $v\not\in U_{\alpha}$. Since $\overline{ v \alpha }$ is the highest node containing $(v,\alpha)$, it suffices to show that $t_{\overline{ v \alpha }}({(v,\alpha))}\ne 1$. But this follows directly from Theorem~\ref{thm:DP}(\ref{asm:DP:req}) because $t_{\overline{ v \alpha }}\in \HH_{\overline{ v \alpha },b[\overline{ v \alpha }]}$, $(v,\alpha)\in X_{\overline{ v \alpha }}$ and  $b[\overline{ v \alpha }]((v,\alpha))\ne1$.
\end{proof}

\paragraph{LP relaxation for  Max-$k$-Cut}
We start with some additional notation related to the tree decomposition ${\cal T}$ (from Theorem~\ref{thm:treedecomp}) and the dynamic program for CSP (from Theorem~\ref{thm:DP}).
\begin{itemize}
\setlength\itemsep{0pt}
	\item For any node $i\in I$, $T_i$ is the set consisting of (1) all nodes $N$ on the $r-i$ path in ${\cal T}$, and  (2) children of all nodes in $N\setminus\{i\}$.
	\item ${\cal P}$ is the collection of all node subsets $J$ such that $J\sse T_{\ell_1} \cup T_{\ell_2}$ for some pair of leaf-nodes $\ell_1,\ell_2$.
	\item $s[i]\in \Sigma_i$ denotes a state at node $i$. Moreover, for any subset of nodes $N\sse I$, we use the shorthand  $s[N]:=\{s[k]:k\in N\}$.
	\item $a[i]\in \Sigma_i$ denotes a state at node $i$ chosen by the algorithm. Similar to $s[N]$, for any subset $N\sse I$ of nodes, $a[N]:=\{a[k]:k\in N\}$.
	\item $\overline{v \alpha}\in I$ denotes the highest tree-decomposition node containing vertex $(v,\alpha)\in V'$.
\end{itemize}

The LP (see Figure~\ref{fig:LP}) that we use here is a generalization of that in~\cite{SLN17}. The variables are $y(s[N])$ for all $\{s[k]\in \Sigma_k\}_{k\in N}$ and $N\in {\cal P}$.  Variable $y(s[N])$ corresponds to the probability of the joint event that the solution (in \gc) ``induces'' state $s[k]$ at each node $k\in N$. Variable $z_{uv\alpha}$ corresponds to the probability that edge $(u,v)\in E$ is cut by part $\alpha$ of the $k$-partition. 

In constraint~\eqref{cons:DP}, we use $j$ and $j'$ to denote the two children of node $i\in I$.
We note that constraints~\eqref{cons:SA}-\eqref{cons:01} which utilize the dynamic-program structure, are identical to the constraints (4)-(8) in the LP from \cite{SLN17}. This allows us to essentially reuse many of the claims proved in \cite{SLN17}, which are stated below.

\begin{figure}[htb]

 \begin{flalign}
&\mbox{maximize } \,\,\frac12\sum_{\{u,v\}\in{V \choose 2}}{c_{uv} \sum_{\alpha = 1}^k z_{uv\alpha}}\tag{LP}\label{LP}&\\
&z_{uv\alpha}\,\, =\,\,  \sum_{\substack{s[\overline{ u \alpha }]\in\Sigma_{\overline{ u \alpha }},\, s[\overline{ v \alpha }]\in\Sigma_{\overline{ v \alpha }} \\ s[\overline{ u \alpha }]((u,\alpha))\ne s[\overline{ v \alpha }]((v,\alpha))}}{y(s[\{\overline{ u \alpha },\overline{ v \alpha }\}])},\notag\\ &\forall \{u,v\}\in{V \choose 2}, \forall \alpha\in[k];\label{cons:obj}&\\
&y(s[N])\,\, =\,\, \sum_{s[i]\in\Sigma_i}{y(s[N\cup\{i\}])},\notag\\
&\forall s[k]\in\Sigma_k,\,\, \forall k\in N,\,\, \forall N\in {\cal P},\,  \forall i\notin N\,:\, N\cup \{i\}\in {\cal P};\label{cons:SA}&\\
&\sum_{s[r]\in\Sigma_r}y(s[r])=1;&\label{cons:r}\\
&y(s[\{i,j,j'\}])\,=\,0,\notag\\
&\forall i\in I, \,  \forall s[i]\in \Sigma_i,\, \forall (s[j],s[j'])\notin {\FF}_{i,s[i]}; \label{cons:DP}&\\
&y(s[\ell])\,=\,0,\notag\\
& \forall \mbox{ leaf }\ell\in I, \, \forall s[\ell]\in \Sigma_\ell : \HH_{\ell,s[\ell]} = \emptyset ; \label{cons:DP-leaf}&\\
&0 \,\,\le \,\,y(s[N])\,\,\le \,\, 1,\notag\\
&\forall N\in {\cal P},\, \forall  s[k]\in\Sigma_k \mbox{ for } k\in N. &\label{cons:01}
\end{flalign}
\caption{The LP formulation}
\label{fig:LP}
\end{figure}

\begin{cl}
	\label{cl:sym-LP-distr}
	Let $y$ be feasible to \eqref{LP}. For any node $i\in I$ with children $j,j'$ and $s[k]\in\Sigma_k$ for all $k\in T_i$,
	\begin{equation*}
	y(s[T_i])\,\,=\,\,\sum_{s[j]\in\Sigma_j}\,\, \sum_{s[j']\in\Sigma_{j'}}{y(s[T_i\cup\{j,j'\}])}.
	\end{equation*} 
\end{cl}
\begin{proof}
Note that $T_i\cup\{j,j'\}\sse T_\ell$ for any leaf node $\ell$ in the subtree below $i$. So $T_i\cup\{j,j'\}\in{\cal P}$ and the variables $y(s[T_i\cup\{j,j'\}])$ are well-defined. The claim 
 follows by two applications of 
 \eqref{cons:SA}.
\end{proof}

\begin{lemma}\label{lem:poly_LP}
	\eqref{LP} has a polynomial number of variables and constraints. 
\end{lemma}
\begin{proof}
	There are ${n\choose 2}\cdot k=O(kn^2)$ variables $z_{uv\alpha}$. Because the tree is binary, we have $|T_i|\le 2d$ for any node $i$, where $d=O(\log n)$ is the depth of the tree decomposition. Moreover there are only $O(n^2)$ pairs of leaves as there are $O(n)$ leaf nodes. For each pair $\ell_1,\ell_2$ of leaves, we have $|T_{\ell_1}\cup T_{\ell_2}|\le 4d$. Thus $|{\cal P}|\le O(n^2)\cdot 2^{4d}=poly(n)$. By Theorem~\ref{thm:DP}, we have $\max|\HH_{i,\sigma}|=O(1)$, so the number of $y$-variables is at most $|{\cal P}|\cdot (\max|\HH_{i,\sigma}|)^{4d}=poly(n)$. This shows that ~\eqref{LP} has polynomial size and can be solved optimally in polynomial time. Finally, it is clear  that the rounding algorithm runs in polynomial time.
\end{proof}

\begin{lemma}\label{lem:LP-LB}
	\eqref{LP} is a valid relaxation of \gcmc.
\end{lemma}
\begin{proof}
Let $k$-partition $\{U_\alpha\}_{\alpha=1}^k$ be any feasible solution to  \gc. Let $\{b[i]\}_{i\in I}$ denote the states given by Claim~\ref{cl:DP-sol-state} corresponding to $\{U_\alpha\}_{\alpha=1}^k$. For any subset $N\in {\cal P}$ of nodes, and for all $\{s[i]\in \Sigma_i\}_{i\in N}$, set
	$$y(s[N]) = \left\{
	\begin{array}{ll}
	1, & \mbox{ if $s[i]=b[i]$ for all }i\in N;\\
	0, & \mbox{ otherwise.}
	\end{array}\right.
	$$
	Clearly constraints~\eqref{cons:SA} and~\eqref{cons:01} are satisfied. By the first property in Claim~\ref{cl:DP-sol-state},
	constraint~\eqref{cons:DP} is satisfied. And by the second property in Claim~\ref{cl:DP-sol-state},
	constraint~\eqref{cons:DP-leaf} is also satisfied. The last property in Claim~\ref{cl:DP-sol-state} implies that $ v\in U_\alpha \iff b[\overline{ v \alpha }]((v,\alpha))=1$ for any vertex $v\in V$. So any edge $\{u,v\}$ is cut by $U_\alpha$ exactly when $b[\overline{ u \alpha }]((u,\alpha))\ne b[\overline{ v \alpha }]((v,\alpha))$. Using the setting of variable $z_{uv\alpha}$ in~\eqref{cons:obj} it follows that $z_{uv\alpha}$ is exactly the indicator of edge $\{u,v\}$ being cut by $U_\alpha$. Finally, the objective value is exactly the total weight of edges cut by the $k$-partition $\{U_\alpha\}_{\alpha=1}^k$  	where the coefficient $\frac12$ comes from the fact that that summation counts each cut-edge twice.  Thus 	\eqref{LP} is a valid relaxation.
	\end{proof}

\paragraph{Rounding Algorithm} This is a top-down procedure, exactly as in \cite{SLN17}.
We start with the root node $r\in I$. Here $\{y(s[r])\,:\,s[r]\in\Sigma_r\}$ defines a probability distribution over the states of $r$. We sample a state $a[r]\in\Sigma_r$ from this distribution. Then we continue top-down: for any node $i\in I$, given the chosen states  $a[k]$ at each $k\in T_i$, we sample states for both children of $i$ {\em simultaneously} from their joint distribution given at node $i$. Our algorithm is formally described in Algorithm \ref{alg:round}.

\begin{algorithm}[htb]\label{alg:round}
	\LinesNumbered
	\SetKwInOut{Input}{Input}\SetKwInOut{Output}{Output}
	\SetKwFor{Do}{Do}{\string:}{end}
	\Input{Optimal solution of \ref{LP}.}
	\Output{A vertex partition of $V$ in ${\cal S}_G$.}
	\label{step:sym-round-sample} Sample a state $a[r]$ at the root node by distribution $y(s[r])$. \\
	\Do{process all nodes $i$ in ${\cal T}$ in order of increasing depth}{
		Sample states $a[j],a[j']$ for the children of node $i$ by joint distribution \label{step:joint-sample}
		\begin{align}
&\Pr[a[j]=s[j]\mbox{ and }a[j']=s[j']]
		 \notag \\
		 &= \quad \frac{y(s[T_i\cup\{j,j'\}])}{y(s[T_i])},\label{eq:sym-SA-round}
\end{align}
		where   $s[T_i]=a[T_i]$.
	}
	\Do{process all nodes $i$ in ${\cal T}$ in order of decreasing depth}{
		$h_i =a[i] \cup h_j\cup h_{j'}$ where $j,j'$ are the children of $i$.
	}
	{Set $U_\alpha = \{ v\in V : h_r((v,\alpha))=1\}$ for all  $\alpha\in [k]$.}\\
	\Return $k$-partition $\{U_\alpha\}_{\alpha=1}^k$.
	\caption{Rounding Algorithm for \ref{LP}}
\end{algorithm}

\begin{lemma}\label{lem:alg-feasible}
	The algorithm's solution $\{U_\alpha\}_{\alpha=1}^k$ is always feasible.
\end{lemma}
\begin{proof}
	Note that the distributions used in Step~\ref{step:sym-round-sample} and Step~\ref{step:joint-sample} are well-defined due to Claim~\ref{cl:sym-LP-distr}; so the states $a[i]$s are well-defined. Moreover, by the choice of these distributions, for each node $i$, $y(a[T_i])>0$.
	
	We now show that  for any node $i\in I$ with children $j,j'$ we have $(a[j],a[j'])\in {\FF}_{i,a[i]}$. Indeed, at the iteration for node $i$ (when $a[j]$ and $a[j']$ are set), using the conditional probability distribution \eqref{eq:sym-SA-round} and   constraint~\eqref{cons:DP}, we have $(a[j],a[j'])\in {\FF}_{i,a[i]}$ with probability one.

	We show by induction that for each node $i\in I$, $h_i\in \HH_{i,a[i]}$. The base case is when $i$ is a leaf. In this case, due to constraint~\eqref{cons:DP-leaf} and the fact that $y(a[T_i])>0$ we know that $\HH_{i,a[i]}\ne \emptyset$. So $h_i=a[i]\in \HH_{i,a[i]}$ by Theorem~\ref{thm:DP}(\ref{asm:DP:leaf}). For the inductive step, consider node $i\in I$ with children $j,j'$ where $h_j\in \HH_{j,a[j]}$ and $h_{j'} \in \HH_{j',a[j']}$.  Moreover, from the property above, $(a[j],a[j'])\in {\FF}_{i,a[i]}$. Now using Theorem~\ref{thm:DP}(\ref{asm:DP:valid}) we have $h_i =a[i] \cup h_{j}\cup h_{j'}\in \HH_{i,a[i]}$. Finally, using $h_r\in \HH_{r,a[r]}$ at the root node and Theorem~\ref{thm:DP}(\ref{asm:DP:root}), it follows that $h_r\in \Feas(\vartheta)$. Now let $h'$ denote the restriction of $h_r$ to the variables $\{(v,\alpha) : v\in V, \alpha\in[k]\}$. Then, using the CSP extension result (Theorem~\ref{thm:MSO_CSP}) we obtain that $h'$ is feasible for $CSP_\varphi(G)$. In other words, the $k$-partition $\{U_\alpha\}_{\alpha=1}^k$ satisfies \gc.	
\end{proof}

\ignore{

\begin{cl}
	Let $h$ be a solution. $h(u)=\alpha$ if and only if $a[\overline{ u \alpha }]((u,\alpha))=1$.
\end{cl}
\begin{proof}
	This proof is identical to that of the last property in Claim~\ref{cl:DP-sol-state}.
	
\end{proof}

}


\begin{cl}
	For any node $i$ and states $s[k]\in \Sigma _k$ for all $k\in T_i$, the rounding algorithm satisfies $\Pr[a[T_i]=s[T_i]]=y(s[T_i])$.\label{cl:path}
\end{cl}
\begin{proof}
We proceed by induction on the depth of node $i$. It is clearly true when $i=r$, i.e. $T_i=\{r\}$. Assuming the statement is true for node $i$, we will prove it for $i$'s children.  Let $j,j'$ be the children nodes of $i$; note that $T_j=T_{j'}=T_i\cup\{j,j'\}$. Then using~\eqref{eq:sym-SA-round}, we have
{\small 	$$\Pr[a[T_j]=s[T_j] \,\, | \,\, a[T_i]=s[T_i] ] = \frac{y(s[T_i\cup\{j,j'\}])}{y(s[T_i])}.$$}
	Combined with $\Pr[a[T_i]=s[T_i]]=y(s[T_i])$ we obtain  $\Pr[a[T_j]=s[T_j]]=y(s[T_j])$ as desired.
\end{proof}

\begin{lemma}\label{lem:sym-alg-cut1}
	Consider any $u,v\in V$ and $\alpha\in [k]$ such that $\overline{u\alpha} \in T_{\overline{v\alpha}}$. Then the probability that edge $(u,v)$ is cut by $U_{\alpha}=\{v\in V:h_r((v,\alpha))=1\}$ is  $z_{uv\alpha}$.
\end{lemma}
\begin{proof}
Applying Claim~\ref{cl:path} with node $i=\overline{v\alpha}$, for any $\{s[k]\in \Sigma _k : k\in T_{\overline{v\alpha}}\}$, we have $\Pr[a[T_{\overline{v\alpha}}] = s[T_{\overline{v\alpha}}]] = y(s[T_{\overline{v\alpha}}])$. Let $D_{u\alpha}=\{s[\overline{u\alpha}]\in\Sigma_{[\overline{u\alpha}]}\, |\, s[\overline{u\alpha}]((u,\alpha))=1\}$ and similarly $D_{v\alpha}=\{s[\overline{v\alpha}]\in\Sigma_{[\overline{v\alpha}]}\, |\, s[\overline{v\alpha}]((v,\alpha))=1\}$. Because $\overline{u\alpha}\in T_{\overline{v\alpha}}$,
{\small \begin{align*}
&\Pr[u\in U_\alpha, v\not\in U_\alpha ] =\sum_{s[\overline{u\alpha}] \in D_{u\alpha}} \sum_{s[\overline{v\alpha}]\not\in D_{v\alpha}} \sum_{\substack{s[k]\in\Sigma_{k} \\ k \in T_{[\overline{v\alpha}]}\setminus \{\overline{u\alpha}\} \setminus \{\overline{v\alpha}\}}}  y(s[\overline{v\alpha}])\\& = \sum_{s[\overline{u\alpha}] \in D_{u\alpha}} \sum_{s[\overline{v\alpha}]\not\in D_{v\alpha}} y(s[\{\overline{u\alpha},\overline{v\alpha}\}]).
\end{align*}}

The last equality above is by repeated application of LP constraint~\eqref{cons:SA} where we use $T_{\overline{v\alpha}}\in {\cal P}$. Similarly,
{\small \begin{align*}
&\Pr[u\not\in U_\alpha, v\in U_\alpha ] = \sum_{s[\overline{u\alpha}]\not \in D_{u\alpha}} \sum_{s[\overline{v\alpha}]\in D_{v\alpha}} y(s[\{\overline{u\alpha},\overline{v\alpha}\}]),
\end{align*}}
which combined with constraint~\eqref{cons:obj} implies $\Pr[|\{u,v\}\cap U_\alpha| =1 ] = z_{uv\alpha}$.
\end{proof}

\begin{lemma}\label{lem:sym-alg-cut2}
	Consider any $u,v\in V$ and $\alpha\in [k]$ such that $\overline{u\alpha} \not\in T_{\overline{v\alpha}}$ and $\overline{v\alpha} \not\in T_{\overline{u\alpha}}$. 	 Then the probability that edge $(u,v)$ is cut by  $U_{\alpha} =\{v\in V:h_r((v,\alpha))=1\}$ is at least $z_{uv\alpha}/2$. 
\end{lemma}
\begin{proof}
We first state a useful observation.

\begin{observation}[Observation 1 in~\cite{SLN17}]
	Let $X,Y$ be two jointly distributed $\{0,1\}$ random variables. Then $\Pr(X=1)\Pr(Y=0)+\Pr(X=0)\Pr(Y=1)\ge\frac12[\Pr(X=0,Y=1)+\Pr(X=1,Y=0)]$.\label{obs:joint_indep}
\end{observation}
Now we start to prove Lemma~\ref{lem:sym-alg-cut2}. In order to simplify notation, we define:
{\small \begin{equation*}
z^+_{uv\alpha}\,\, =\,\, \sum_{\substack{s[\overline{ u \alpha }]\in\Sigma_{\overline{ u \alpha }},\, s[\overline{ v \alpha }]\in\Sigma_{\overline{ v \alpha }} \\ s[\overline{ u \alpha }]((u,\alpha))=1, s[\overline{ v \alpha }]((v,\alpha))=0}}{y(s[\{\overline{ u \alpha },\overline{ v \alpha }\}])},
\end{equation*}}
{\small \begin{equation*}
z^-_{uv\alpha}\,\, =\,\, \ \sum_{\substack{s[\overline{ u \alpha }]\in\Sigma_{\overline{ u \alpha }},\, s[\overline{ v \alpha }]\in\Sigma_{\overline{ v \alpha }} \\ s[\overline{ u \alpha }]((u,\alpha))=0, s[\overline{ v \alpha }]((v,\alpha))=1}}{y(s[\{\overline{ u \alpha },\overline{ v \alpha }\}])}.
\end{equation*}}
Note that $z_{uv} = z_{uv\alpha}^+ + z_{uv\alpha}^-$.

Let $D_{u\alpha}=\{s[\overline{u\alpha}]\in\Sigma_{[\overline{u\alpha}]}\, |\, s[\overline{u\alpha}]((u,\alpha))=1\}$ and $D_{v\alpha}=\{s[\overline{v\alpha}]\in\Sigma_{[\overline{v\alpha}]}\, |\, s[\overline{v\alpha}]((v,\alpha))=1\}$. Let $i$ denote the least common ancestor of nodes $\overline{u\alpha}$ and $\overline{v\alpha}$, and $\{j,j'\}$ the two children of $i$. Note that $T_j = T_{j'} = T_i\cup\{j,j'\}$ and $T_{\overline{u\alpha}} , T_{\overline{v\alpha}} \supseteq T_j$. Because $\overline{u\alpha} \not\in T_{\overline{v\alpha}}$ and $\overline{v\alpha} \not\in T_{\overline{u\alpha}}$, both $\overline{u\alpha}$ and $\overline{v\alpha}$ are strictly below $j$ and $j'$ (respectively) in the tree decomposition.

For any choice of states $\{s[k]\in \Sigma_k\}_{k\in T_j}$ define:
{\small $$z^+_{uv\alpha}(s[T_j]) = \sum_{s[\overline{u\alpha}] \in D_{u\alpha}} \sum_{s[\overline{v\alpha}]\not\in D_{v\alpha}} \frac{y(s[T_j\cup\{\overline{u\alpha},\overline{v\alpha}\}])}{y(s[T_j])},$$}
and similarly $z^-_{uv\alpha}(s[T_j])$.

In the rest of the proof, we fix states $\{s[k]\in \Sigma_k\}_{k\in T_j}$ and {\em condition} on the event ${\cal E}$ that $a[T_j]=s[T_j]$. We will show $\Pr[|\{u,v\}\cap U_\alpha|=1 \, |\, {\cal E}]$
{\small \begin{equation}\label{eq:sym-cut-prob}
\,\, \ge \,\, \frac12\left(z^+_{uv\alpha}(s[T_j])+z^-_{uv\alpha}(s[T_j])\right).
\end{equation}
}

By taking expectation over the conditioning ${\cal E}$, this would imply  Lemma~\ref{lem:sym-alg-cut2}.

We now define the following indicator random variables (conditioned on ${\cal E}$).
{\small \begin{equation*}
I_{u\alpha}=\begin{cases}
0\quad \mbox{ if }a[\overline{u\alpha}] \not\in D_{u\alpha}\\
1\quad \mbox{ if }a[\overline{u\alpha}] \in D_{u\alpha}
\end{cases}
\mbox{and }\,\,
I_{v\alpha}=\begin{cases}
0\quad \mbox{ if }a[\overline{v\alpha}]\not\in D_{v\alpha}\\
1\quad \mbox{ if }a[\overline{v\alpha}] \in D_{v\alpha}
\end{cases}
\end{equation*} }
Observe that $I_{u\alpha}$ and $I_{v\alpha}$ ({\em conditioned} on ${\cal E}$) are independent because $\overline{u\alpha}, \overline{v\alpha}\not\in T_{j}$, and $\overline{u\alpha}$ and $\overline{v\alpha}$ appear in distinct subtrees under node $i$. So,
{\small \begin{align}
&\Pr[|\{u,v\}\cap U_\alpha|=1 \, |\, {\cal E}]\notag\\& = \Pr[I_{u\alpha}=1]
 \cdot \Pr[I_{v\alpha}=0] + \Pr[I_{u\alpha}=0]\cdot \Pr[I_{v\alpha}=1]\label{eq:sym-alg-cut-prob}
\end{align}}

For any $s[k]\in\Sigma_k$ for $k\in T_{\overline{u\alpha}}\setminus T_j$, we have by Claim \ref{cl:path} and $T_j\sse T_{\overline{u\alpha}}$ that
{\small \begin{align*}
&\Pr[a[T_{\overline{u\alpha}}]=s[T_{\overline{u\alpha}}]\,|\, a[T_j]=s[T_j]] \,\, =\,\, \frac{y(s[T_{\overline{u\alpha}}])}{y(s[T_j])}.
\end{align*}}
Therefore $\Pr[I_{u\alpha}=1]$ equals
{\small \begin{align*}
&\sum_{s[\overline{u\alpha}]\in D_{u\alpha}} \,\, \sum_{\substack{k\in T_{\overline{u\alpha}}\setminus T_j\setminus \{\overline{u\alpha}\} \\s[k]\in\Sigma_k}}\frac{y(s[T_{\overline{u\alpha}}])}{y(s[T_j])} =\sum_{s[\overline{u\alpha}]\in D_{u\alpha}}\frac{y(s[T_{j}\cup\{\overline{u\alpha}\}])}{y(s[T_j])}.
\end{align*} }
The last equality follows by repeatedly using LP constraint~\eqref{cons:SA} and the fact that $T_{\overline{u\alpha}}\in {\cal P}$. Furthermore, note that $T_j\cup\{\overline{u\alpha},\overline{v\alpha}\}\in {\cal P}$; again by constraint~\eqref{cons:SA},
{\small \begin{align*}
&\Pr[I_{u\alpha}=1]=\sum_{s[\overline{u\alpha}]\in D_{u\alpha}}\frac{y(s[T_{j}\cup\{\overline{u\alpha}\}])}{y(s[T_j])}\\
&\qquad=\sum_{s[\overline{u\alpha}]\in D_{u\alpha}}\sum_{s[\overline{v\alpha}]\in \Sigma_{\overline{v\alpha}}}\frac{y(s[T_{j}\cup\{\overline{u\alpha},\overline{v\alpha}\}])}{y(s[T_j])}\\
&\qquad=\sum_{s[\overline{u\alpha}]\in D_{u\alpha}}\sum_{s[\overline{v\alpha}]\in D_{v\alpha}}\frac{y(s[T_{j}\cup\{\overline{u\alpha},\overline{v\alpha}\}])}{y(s[T_j])} + z_{uv\alpha}^+(s[T_j])
\end{align*}
}

Similarly,
{\small \begin{align*}
	&\Pr[I_{v\alpha}=1]=\sum_{s[\overline{v\alpha}]\in D_{v\alpha}}\frac{y(s[T_{j}\cup\{\overline{v\alpha}\}])}{y(s[T_j])}\\
	&\qquad=\sum_{s[\overline{u\alpha}]\in D_{u\alpha}}\sum_{s[\overline{v\alpha}]\in D_{v\alpha}}\frac{y(s[T_{j}\cup\{\overline{u\alpha},\overline{v\alpha}\}])}{y(s[T_j])} + z_{uv\alpha}^-(s[T_j]).
\end{align*}
\begin{align*}
&\Pr[I_{u\alpha}=0]= \sum_{s[\overline{u\alpha}]\not\in D_{u\alpha}}\frac{y(s[T_{j}\cup\{\overline{u\alpha}\}])}{y(s[T_j])}\\
&\qquad=\sum_{s[\overline{u\alpha}]\not\in D_{u\alpha}}\sum_{s[\overline{v\alpha}]\not\in D_{v\alpha}}\frac{y(s[T_{j}\cup\{\overline{u\alpha},\overline{v\alpha}\}])}{y(s[T_j])} + z_{uv\alpha}^-(s[T_j]).
\end{align*}
}

Now define $\{0,1\}$ random variables $X$ and $Y$ jointly distributed as:
\begin{table}[h!]
	\centering
	\resizebox{\columnwidth}{!}{
		\begin{tabular}{|c|c|c|}
			\hline
			&$Y=0$ & $Y=1$\\
			\hline
			$X=0$ & $\Pr[I_{u\alpha}=0] - z_{uv\alpha}^-(s[T_j])$ & $z_{uv\alpha}^-(s[T_j])$  \\
			\hline
			$X=1$ & $z_{uv\alpha}^+(s[T_j])$ & $\Pr[I_{u\alpha}=1] - z_{uv\alpha}^+(s[T_j])$\\
			\hline
		\end{tabular}
	}
\end{table}

Note that $\Pr[X=1]=\Pr[I_{u\alpha}=1]$ and $\Pr[Y=1]=\Pr[I_{u\alpha}=1] - z_{uv\alpha}^+(s[T_j])+z_{uv\alpha}^-(s[T_j])=\Pr[I_{v\alpha}=1]$. So, applying Observation~\ref{obs:joint_indep} and using~\eqref{eq:sym-alg-cut-prob} we have $\Pr[|\{u,v\}\cap U_\alpha|=1 \, |\, {\cal E}]$ is at least
{\small \begin{align*}
\frac12\left( \Pr[X=0,Y=1]+\Pr[X=1,Y=0]\right),
\end{align*} }
which implies~\eqref{eq:sym-cut-prob}.
\end{proof}

\begin{lemma}\label{lem:alg-cut}
	For any $u,v\in V$, the probability that edge $(u,v)$ is cut by the $k$-partition $\{U_\alpha\}_{\alpha=1}^k$ is at least $\frac14 \sum_{\alpha=1}^k z_{uv\alpha}$.
\end{lemma}
\begin{proof}
Edge	$(u,v)$ is cut by $\{U_\alpha\}_{\alpha=1}^k$  if and only if $u\in U_{\alpha}$ and $v\in U_{\beta}$ for some $\alpha\ne \beta$. Enumerating all partition parts and applying Lemmas~\ref{lem:sym-alg-cut1} and~\ref{lem:sym-alg-cut2}, we get that the probability is at least $\frac14 \sum_{\alpha=1}^k z_{uv\alpha}$. The extra factor of $\frac{1}{2}$ is because any cut edge $(u,v)$ is cut by the partition twice: by the parts containing $u$ and $v$. 
\end{proof}
\noindent From Lemmas~\ref{lem:LP-LB}, \ref{lem:alg-feasible} and \ref{lem:alg-cut}, we obtain Theorem~\ref{thm:gen-k}.

\section{Applications}
We claim that \MSOD is powerful enough to model various graph properties; to that end, consider the following formulae, meant to model that a set $S$ is a vertex cover, an independent set, a dominating set, and a connected set, respectively:
\begin{align*}
&\varphi_{\text{vc}}(S) \equiv \forall \{u,v\} \in E:\, (u \in S) \vee (v \in S) \\
&\varphi_{\text{is}}(S) \equiv \forall \{u,v\} \in E:\, \neg \left((u \in S) \wedge (v \in S) \right) \\
&\varphi_{\text{ds}}(S) \equiv \forall v \in V:\, \exists u \in S:\, (v \not\in S) \\
& \qquad \qquad \qquad \qquad \implies \{u,v\} \in E \\
&\varphi_{\text{conn}}(S) \equiv \neg \big[  \exists U,V \subseteq S:\, U \cap V = \emptyset \wedge U \cup V\\
& \qquad = S   \wedge \neg \big( \exists \{u,v\} \in E: u \in U \wedge v \in V \big) \big]
\end{align*}
We argue as follows: $\varphi_{\text{vc}}$ is true if every edge has at least one endpoint in $S$; $\varphi_{\text{is}}$ is true if every edge does not have both endpoints in $S$; $\varphi_{\text{ds}}$ is true if for each vertex $v$ not in $S$ there is a neighbor $u$ in $S$; finally, $\varphi_{\text{conn}}$ is true if there does not exist a partition $U,V$ of $S$ with an edge going between $U$ and $V$.

We also show how to handle the precedence constraint.
Let $G$ be a directed graph; we require $S$ to satisfy that, for each arc $(u,v) \in E$, either $v \not\in S$, or $u, v \in S$.
This can be handled directly with CSP constraints: we have a binary variable for each vertex with the value $1$ indicating that a vertex is selected for $S$; then, for each arc $(u,v) \in E$, we have a constraint $C_{(u,v)} = \{(1,0), (0,0), (1,1)\}$.

It is known~\cite{Courcelle:1996} that many other properties are expressible in \MSOD, such as that $S$ is $k$-colorable, $k$-connected (both for fixed $k \in \NN$), planar, Hamiltonian, chordal, a tree, not containing a list of graphs as minors, etc.
It is also known how to encode directed graphs into undirected graphs in an ``\MSO-friendly'' way~\cite{Courcelle:1996}, which allows the expression of various properties of directed graphs.
Our results also extend to so-called \emph{counting \MSO}, where we additionally have a predicate of the form $|X| = p~\textrm{mod}~q$ for a fixed integer $q \in \NN$.

\section{Conclusions}
In this paper we obtained  $\frac12$-approximation algorithms for graph-\MSO-constrained max-$k$-cut problems, where the constraint graph has bounded treewidth. This work generalizes the class of constraints handled in~\cite{SLN17} and extends the result to the setting of max-$k$-cut. Getting an approximation ratio better than $\frac12$ for any of these problems is an interesting question, even for a specific \MSO-constraint.
Regarding Remark~\ref{rem:pc}, could our algorithm be improved to an FPT algorithm (runtime $g(\tau) n^{O(1)}$ for some function $g$)?
If not, is there an FPT algorithm parameterized by the (more restrictive) tree-depth of $G$?

\section*{Acknowledgement}
We thank Samuel Fiorini for raising the possibility of a systematic extension of our prior work \cite{SLN17}, which eventually lead to this paper.

Research of M. Koutecký is supported by a Technion postdoc grant. Research of J. Lee is supported in part by ONR grant N00014-17-1-2296. Part of this work was done while J. Lee was visiting the Simons Institute for the Theory of Computing (which was partially supported by the DIMACS/Simons Collaboration on Bridging Continuous and Discrete Optimization through NSF grant CCF-1740425). Research of V. Nagarajan and X. Shen is supported in part by NSF CAREER grant CCF-1750127.
\section*{References}
\bibliographystyle{elsarticle-num}
\bibliography{gcmc}

\end{document}